\documentclass[12pt]{article}

\usepackage{times}
\usepackage{amsthm,amsmath,amssymb,array,latexsym,booktabs,graphicx,multirow}
\usepackage{bm}
\usepackage{natbib}
\usepackage{color}
\usepackage[T1]{fontenc}
\usepackage[utf8]{inputenc}
\usepackage{authblk}
\usepackage{geometry}
\allowdisplaybreaks[3]

\renewcommand{\baselinestretch}{2}

\newtheorem{algorithm}{Algorithm}
\newtheorem{lemma}{Lemma}

 \newtheorem{theorem}{Theorem}
 \newtheorem{proposition}{Proposition}

\newcommand{\pr}{\mathrm{Pr}}
\def\T{{ \mathrm{\scriptscriptstyle T} }}

\newcommand{\bSigma}{{\bm \Sigma}}
\newcommand{\bX}{\mathbf{X}}
\newcommand{\bY}{\mathbf{Y}}
\newcommand{\bx}{\mathbf{x}}

\newcommand{\bbeta}{{\bm\beta}}

\newcommand{\bA}{\mathbf{A}}
\newcommand{\bB}{\mathbf{B}}

\newcommand{\bI}{\mathbf{I}}

\newcommand{\bu}{\mathbf{u}}

\newcommand{\bt}{\mathbf{t}}
\newcommand{\by}{\mathbf{y}}

\newcommand{\bdelta}{{\bm \delta}}
\newcommand{\balpha}{{\bm \alpha}}
\newcommand{\btheta}{{\bm \theta}}
\newcommand{\bmeta}{{\bm \eta}}

\newcommand{\bOmega}{{\bm \Omega}}

\newcommand{\bmu}{\bm \mu}

\newcommand{\bPi}{\bm \Pi}

\def\T{{ \mathrm{\scriptscriptstyle T} }}

\def\pr{{ \mathrm{pr} }}

\newcommand{\beq}{\begin{equation}}
\newcommand{\eeq}{\end{equation}}
\newcommand{\beqn}{\begin{equation*}}
\newcommand{\eeqn}{\end{equation*}}
\newcommand{\bea}{\begin{eqnarray}}
\newcommand{\eea}{\end{eqnarray}}
\newcommand{\bean}{\begin{eqnarray*}}
\newcommand{\eean}{\end{eqnarray*}}
\newcommand{\vinf}[1]{\Vert #1 \Vert_{\infty}}
\newcommand{\vone}[1]{\Vert #1 \Vert_{1}}

\begin{document}
\title{Multiclass Sparse Discriminant Analysis}
\author{Qing Mai \thanks{Department of Statistics, Florida State University (mai@stat.fsu.edu)},
Yi Yang \thanks{School of Statistics, University of Minnesota (yiyang@umn.edu). Mai and Yang are joint first authors.},
Hui Zou \thanks{Corresponding author, School of Statistics, University of Minnesota (zouxx019@umn.edu)}}

\date{First version: May 30, 2014\\
This version: April 17, 2015}

\maketitle

\begin{abstract}
In recent years many sparse linear discriminant analysis methods have been proposed for high-dimensional classification and variable selection. However, most of these proposals focus on binary classification and they are not directly applicable to multiclass classification problems. There are two sparse discriminant analysis methods that can handle multiclass classification problems, but their theoretical justifications remain unknown. In this paper, we propose a new multiclass sparse discriminant analysis method that estimates all discriminant directions simultaneously. We show that when applied to the binary case our proposal yields a classification direction that is equivalent to those by two successful binary sparse LDA methods in the literature. An efficient algorithm is developed for computing our method with high-dimensional data. Variable selection consistency and rates of convergence are established under the ultrahigh dimensionality setting. We further demonstrate the superior performance of our proposal over the existing methods on simulated and real data.
\end{abstract}

Keywords:
Discriminant analysis; High dimensional data; Variable selection; Multiclass classification; Rates of convergence.

\section{Introduction}
 In multiclass classification we have a pair of random variables $(Y,\bX)$, where $\bX\in \mathbb{R}^p$ and $Y\in \{1,\ldots, K\}$. We need to predict $Y$ based on   $\bX$. Define  $\pi_k=\Pr(Y=k)$. The linear discriminant analysis model states that
\begin{equation}\label{LDA}
\bX\mid (Y=k)\sim N(\bmu_k,\bSigma), k\in \{1,2,\ldots,K\}.
\end{equation}
 Under \eqref{LDA}, the Bayes rule can be explicitly derived as follows
\beq\label{Bayes}
\hat Y=\arg\max_k \{(\bX-\dfrac{\bmu_k}{2})^\T\bbeta_k+\log{\pi_k} \},
\eeq
where $\bbeta_k=\bSigma^{-1}\bmu_k$ for $k=1,\ldots, K$.
Linear discriminant analysis has been observed to perform very well on many low-dimensional datasets \citep{STATLOG,Hand_2006}. However, it may not be suitable for high-dimensional datasets for at least two reasons. First, it is obvious that linear discriminant analysis cannot be applied if the dimension $p$ exceeds the sample size $n$, because the sample covariance matrix will be singular. Second, \cite{FF2008} showed that even if the true covariance matrix is an identity matrix and we know this fact, a classifier involving all the predictors will be no better than random guessing.

In recent years, many high-dimensional generalizations of linear discriminant analysis have been proposed  \citep{Tibshirani_shrunken,TJ,CLT_2008,FF2008,Wu08,Shao,CL2011,WT_2011,Slda,ROAD}. In the binary case, the discriminant direction is $\bbeta=\bSigma^{-1}(\bmu_2-\bmu_1)$. One can seek sparse estimates of $\bbeta$  to generalize linear discriminant analysis to deal with high dimensional classification. Indeed, this is the common feature of three popular sparse discriminant analysis methods: the linear programming discriminant \citep{CL2011}, the regularized optimal affine discriminant \citep{ROAD} and the direct sparse discriminant analysis \citep{Slda}. The linear programming discriminant finds a sparse estimate by the Dantzig selector \citep{Dantzig}; the regularized optimal affine discriminant \citep{ROAD} adds the lasso penalty \citep{Tibs96a} to Fisher's discriminant analysis; and  the direct sparse discriminant analysis \citep{Slda} derives the sparse discriminant direction via a sparse penalized least squares formulation.  The three methods can detect the important predictors and consistently estimate the classification rule with overwhelming probabilities with the presence of ultrahigh dimensions. However, they are explicitly designed for binary classification and do not handle the multiclass case naturally.

Two popular multiclass sparse discriminant analysis proposals are the $\ell_1$ penalized Fisher's discriminant \citep{WT_2011} and sparse optimal scoring \cite{CLT_2008}. However, these two methods do not have theoretical justifications. It is generally unknown whether they can select the true variables with high probabilities, how close their estimated discriminant directions are to the true directions, and whether the final classifier will work similarly as the Bayes rule. 

Therefore, it is desirable to have a new multiclass sparse discriminant analysis algorithm that is conceptually intuitive, computationally efficient and theoretically sound. To this end, we propose a new sparse discriminant method for high-dimensional multiclass problems. We show that our proposal not only has competitive empirical performance but also enjoys strong theoretical properties under ultrahigh dimensionality.  In Section 2 we introduce the details of our proposal after briefly reviewing the existing two proposals. We also develop an efficient algorithm for our method. Theoretical results are given in Section 3. In Section 4 we use simulations and a real data example to demonstrate the superior performance of our method over sparse optimal scoring \citep{CLT_2008} and $\ell_1$ penalized Fisher's discriminant \citep{WT_2011}. Technical proofs are in an Appendix.

\section{Method}

\subsection{Existing proposals}

The Bayes rule under a linear discriminant analysis model is
\beqn\label{Bayes}
\hat Y=\arg\max_k \{(\bX-\dfrac{\bmu_k}{2})^\T\bbeta_k+\log{\pi_k} \},
\eeqn
where $\bbeta_k=\bSigma^{-1}\bmu_k$ for $k=1,\ldots, K$. Let $\btheta_k^{\mathrm{Bayes}}=\bbeta_{k}-\bbeta_{1}$ for $k=1,\ldots,K$. Then the Bayes rule can be written as
\beq\label{Bayes2}
\hat Y=\arg\max_k\{(\btheta_k^{\mathrm{Bayes}})^\T(\bX-\dfrac{\bmu_k}{2})+\log{\pi_k}\}.
\eeq
We refer to the directions $\btheta^{\mathrm{Bayes}}=(\btheta_2^{\mathrm{Bayes}},\ldots,\btheta_K^{\mathrm{Bayes}})\in\mathbb{R}^{p\times(K-1)}$ as the discriminant directions.

We briefly review two existing multiclass sparse discriminant methods: the sparse optimal scoring \citep{CLT_2008} and the $\ell_1$ penalized Fisher's discriminant \citep{WT_2011}.
Instead of estimating $\btheta^{\mathrm{Bayes}}$ directly, these two methods estimate a set of directions $\bmeta=(\bmeta_1,\ldots,\bmeta_{K-1})\in\mathbb{R}^{p\times(K-1)}$ such that $\bmeta$ spans the same linear subspace as $\btheta^{\mathrm{Bayes}}$ and hence linear discriminant analysis on $\bX^\T\bmeta$ will be equivalent to \eqref{Bayes2} on the population level. More specifically, 
these two methods look for estimates of $\bmeta=(\bmeta_1,\ldots,\bmeta_{K-1})$ in Fisher's discriminant analysis:
\begin{equation}\label{fisher}
\bmeta_k=\arg\max \bmeta_k^\T\bSigma_b\bmeta_k, \mbox{s.t. $\bmeta_k^\T\bSigma\bmeta_k=1, \bmeta_k^\T\bSigma\bmeta_l=0$ for $l<k$,}
\end{equation}
where $\bSigma_b=\frac{1}{K-1}\sum_{k=1}^{K}(\bmu_k-\bar\bmu)(\bmu_k-\bar\bmu)^\T$ with $\bar{\bmu}=\frac{1}{K}\sum_{k}\bmu_k$.

 With a little abuse of terminology, we refer to $\bmeta$ as discriminant directions as well. To find $\bmeta$, define $\bY^{\mathrm{dm}}$ as an $n\times K$ matrix of dummy variables with $Y^{\mathrm{dm}}_{ik}=\mathrm{1}(Y_i=k)$.

In addition to the discriminant direction $\bmeta_k$, sparse optimal scoring creates $K-1$ vectors of scores $\balpha_1,\ldots, \balpha_{K-1}\in \mathbb{R}^{K}$. Then for $k=1,\ldots,K-1$, sparse optimal scoring estimates $\bmeta_k$ sequentially. In each step, sparse optimal scoring finds $\hat\balpha_k, \hat\bmeta^{\mathrm{SOS}}_k$. Suppose the first $k-1$ score vectors $\hat\balpha_{l},\ l<k$ and discriminant directions $\hat\bmeta_{l}^{\mathrm{SOS}},\ l<k$ are available. Then sparse optimal scoring finds $\hat\balpha_k,\hat\bmeta_k^{\mathrm{SOS}}$ by solving the following problem:
\begin{eqnarray}\label{sb9.7.eq2}
(\hat\balpha_k,\hat\bmeta_k^{\mathrm{SOS}})&=&\arg\min_{\balpha_k,\bmeta_k}\sum_{i=1}^n(\mathbf{Y}^{\mathrm{dm}}\balpha_k-\tilde{\mathbf{X}}\bmeta_k)^2+\lambda\Vert\bmeta_k\Vert_1\\
&&\mbox{ s.t. } \dfrac{1}{n}\balpha_k^\T (\mathbf{Y}^{\mathrm{dm}})^\T
\mathbf{Y}^{\mathrm{dm}}\balpha_k=1, \balpha_k^\T (\mathbf{Y}^{\mathrm{dm}})^\T \mathbf{Y}^{\mathrm{dm}}\hat\balpha_l=0, \mbox{ for any }l<k,\nonumber
\end{eqnarray}
where $\tilde\bX$ is the centered data matrix, and $\lambda$ is a tuning parameter. The sparse optimal scoring is closely related to \eqref{fisher}, because when the dimension is low, the unpenalized version of \eqref{sb9.7.eq2} gives the same directions (up to a scalar) as \eqref{fisher} with the parameters $\bSigma_b$ and $\bSigma$ substituted with the sample estimates. Therefore, with the $\ell_1$ penalty, sparse optimal scoring gives sparse approximations to $\bmeta$.

Note that the constraint $\balpha_k^\T (\mathbf{Y}^{\mathrm{dm}})^\T \mathbf{Y}^{\mathrm{dm}}\balpha_l=0, l<k$ indicates that, $(\hat\balpha_k,\hat\bmeta_k^{\mathrm{SOS}})$ depends on the knowledge of $(\hat\balpha_l,\hat\bmeta_l^{\mathrm{SOS}}), l<k$. This is why we say that the sparse optimal scoring adopts a sequential approach to estimate the discriminant directions.

The $\ell_1$ penalized Fisher's discriminant analysis estimates $\bmeta_k$ by
\begin{equation*}\label{Fisher}
\hat\bmeta_k=\arg\max_{\bmeta_k}\bmeta_k^\T\hat\bSigma^k_b\bmeta_k+\lambda_k\sum_{j}|\hat\sigma_j\eta_{kj}| \mbox{
s.t. } \bmeta_k^\T\tilde\bSigma\bmeta_k\le 1,
\end{equation*}
for $k=1,\ldots, K-1$, where $\lambda_k$ are tuning parameters, $\hat\sigma_j^2$ is the $(j,j)$th element of the sample estimate of $\bSigma$, $\tilde\bSigma$ is a positive definite estimate of $\bSigma$,
\begin{equation}\label{sb9.5.eq2}
\hat\bSigma_b^k=\mathbf{X}^\T \mathbf{Y}^{\mathrm{dm}}((\mathbf{Y}^{\mathrm{dm}})^\T\mathbf{
Y}^{\mathrm{dm}})^{-1/2}\bOmega_k((\mathbf{Y}^{\mathrm{dm}})^\T \mathbf{Y}^{\mathrm{dm}})^{-1/2}(\mathbf{Y}^{\mathrm{dm}})^\T
\mathbf{X}
\end{equation}
and $\bOmega_k$ is the identity matrix if $k=1$ and otherwise an orthogonal projection matrix with column space orthogonal to $((\bY^{\mathrm{dm}})^\T \bY)^{-1/2}\bY^\T \bX\hat\bmeta_l$ for
all $l<k$. Again, if the dimension is low, then unpenalized version of \eqref{sb9.5.eq2} is equivalent to \eqref{fisher} with the parameters replaced by the sample estimates. Since $\bOmega_k$ relies on $\hat \bmeta_l$ for all $l<k$, the $\ell_1$ penalized Fisher's discriminant analysis also finds the discriminant directions sequentially.

\subsection{Our proposal}

Good empirical results have been reported for supporting the $\ell_1$ penalized Fisher's discriminant analysis and the sparse optimal scoring.
However, it is unknown whether either of these two classifiers is consistent when more than two classes are present. Moreover,
both sparse optimal scoring and $\ell_1$ penalized Fisher's discriminant analysis estimate the discriminant directions sequentially.
We believe a better multiclass sparse discriminant analysis algorithm should be able to estimate all discriminant directions simultaneously, just like the classical linear discriminant analysis. We aim to develop a new computationally efficient multiclass sparse discriminant analysis method that enjoy strong theoretical properties under ultrahigh dimensionality. Such a method can be viewed as a natural multiclass counterpart of the three binary sparse discriminant methods in \cite{Slda}, \cite{CL2011} and \cite{ROAD}.

To motivate our method, we first discuss the implication of sparsity in the multiclass problem. Note that, by \eqref{Bayes2}, the contribution from the $j$th variable ($X_j$) will vanish if and only if
\beq\label{condition2}
\theta^{\mathrm{Bayes}}_{2j}=\cdots=\theta^{\mathrm{Bayes}}_{Kj}=0
\eeq
Let ${\cal D}=\{j: \textrm{condition}~(\ref{condition2}) \ \textrm{does not hold} \}$. Note that whether an index $j$ belongs to $\cal D$ depends on $\theta_{kj}$ for all $k$. This is because $\theta_{kj}^{\mathrm{Bayes}}, k=2,\ldots,K$ are related to each other, as they are coefficients for the same predictor. In other words, $\theta_{kj}^{\mathrm{Bayes}}, k=2,\ldots, K$ are naturally grouped according to $j$. Then the sparsity assumption states that $|{\cal D}| \ll p$, which is referred to as the common sparsity structure.

Our proposal begins with a convex optimization formulation of the Bayes rule of the multiclass linear discriminant analysis model.
Recall that $\btheta_k^{\mathrm{Bayes}}=\bSigma^{-1}(\bmu_k-\bmu_1)$ for $k=2,\ldots,K$.  On the population level, we have
\begin{eqnarray} \label{convex1}
(\btheta_2^{\mathrm{Bayes}},\ldots,\btheta_K^{\mathrm{Bayes}})=\arg\min_{\btheta_2,\ldots,\btheta_K} \sum^{K}_{k=2}\{\frac{1}{2}\btheta^\T_k \bSigma \btheta_k-( \bmu_k- \bmu_1)^{\T} \btheta_k\}.
\end{eqnarray}
In the classical low-dimension-large-sample-size setting, we can estimate $(\btheta_2^{\mathrm{Bayes}},\ldots,\btheta_K^{\mathrm{Bayes}})$ via an empirical version of (\ref{convex1})
\begin{eqnarray}\label{LDA1}
(\hat\btheta_2,\ldots,\hat\btheta_K)=\arg\min_{\btheta_2,\ldots,\btheta_K} \sum^{K}_{k=2}\{\frac{1}{2}\btheta^\T_k \hat\bSigma \btheta_k-(\hat\bmu_k-\hat\bmu_1)^{\T} \btheta_k\},
\end{eqnarray}
where $\hat\bSigma=\dfrac{1}{n-K}\sum_{k=1}^K\sum_{Y^i=k}(\bX^i-\hat\bmu_k)(\bX^i-\hat\bmu_k)^\T$, $\hat\bmu_k=\dfrac{1}{n_k}\sum_{Y^i=k}\bX^i$ and $n_k$ is the sample size within Class $k$. The solution to (\ref{LDA1}) gives us the classical multiclass linear discriminant classifier.

For presentation purpose, write
$
\btheta_{.j}=(\theta_{2j},\ldots,\theta_{Kj})^{\T}
$
and define
$
\Vert\btheta_{.j} \Vert=(\sum^K_{i=2} \theta_{ij}^2)^{1/2}.
$
For the high-dimensional case, we propose the following penalized formulation for multiclass sparse discriminant analysis.
\begin{eqnarray}\label{estimator2}
(\hat\btheta_2,\ldots,\hat\btheta_K)=\arg\min_{\btheta_2,\ldots,\btheta_K} \sum^{K}_{k=2}\{\frac{1}{2}\btheta^\T_k \hat \bSigma \btheta_k-(\hat \bmu_k-\hat \bmu_1)^{\T} \btheta_k\}+\lambda \sum_{j=1}^p\Vert\btheta_{\cdot j}\Vert,
\end{eqnarray}
where $\lambda$ is a tuning parameter. It is clear that  (\ref{estimator2}) is based on  (\ref{LDA1}).
 In (\ref{estimator2}) we have used the group lasso \citep{groupedlasso} to encourage the common sparsity structure.
Let $\hat{\cal D}=\{j: \hat\theta_{kj} \ne 0\}$ which denotes the set of selected variables for the multiclass classification problem.
We will show later that with a high probability $\hat {\cal D}$ equals ${\cal D}$.
One can also use a group version of a nonconvex penalty \citep{FL_2001} or an adaptive group lasso penalty \citep{bach08} to replace the group lasso penalty in \eqref{estimator2}. To fix the main idea, we do not pursue this direction here.

After obtaining $\hat \btheta_k,k=2,\ldots, K$, we fit the classical multiclass linear discriminant analysis on $(\bX^\T\hat \btheta_2,\ldots,\bX^\T\hat\btheta_{K})$, as in sparse optimal scoring and $\ell_1$ penalized Fisher's discriminant analysis. We repeat the procedure for a sequence of $\lambda$ values and pick the one with the smallest cross-validation error rate.

We would like to make a remark here that our proposal is derived from a different angle than sparse optimal scoring and $\ell_1$ penalized Fisher's discriminant analysis. Both sparse optimal scoring and $\ell_1$ penalized Fisher's discriminant analysis penalize a formulation related to Fisher's discriminant analysis in \eqref{fisher}, while our method directly estimates the Bayes rule. This different angle leads to considerable convenience in both computation and theoretical studies. Yet we can easily recover the directions defined by Fisher's discriminant analysis after applying our method. See Section A.1 for details.

\subsection{Connections with existing binary sparse LDA methods}
Although our proposal is primarily motivated by the multiclass classification problem, it can be directly applied to the binary classification problem as well by simply letting $K=2$ in the formulation \eqref{estimator2}. It turns out that the binary special case of our proposal has very intimate connections with some proven successful binary sparse LDA methods in the literature. We elaborate more on this point in what follows. 

When $K=2$, \eqref{estimator2} reduces to
\begin{equation}
\hat\btheta^{\mathrm{MSDA}}(\lambda)=\arg\min_{\btheta} \frac{1}{2}\btheta^\T \hat \bSigma \btheta-(\hat \bmu_2-\hat \bmu_1)^{\T} \btheta+\lambda \Vert\btheta\Vert_1\label{MSDA11}
\end{equation}
Considering the Dantzig selector formulation of (\ref{MSDA11}), we have the following constrained $\ell_1$ minimization estimator defined as
\begin{equation}
\hat\btheta=\arg\min_{\btheta} \Vert\btheta\Vert_1 \mbox{ s.t. $\Vert\hat\bSigma\btheta-(\hat\bmu_2-\hat\bmu_1)\Vert_{\infty}\le\lambda$}.
\end{equation}
The above estimator is exactly the linear programming discriminant (LPD) \cite{CL2011}. 
 
Moreover, we compare (\ref{MSDA11}) with another two well-known sparse discriminant analysis proposals for binary classification: the regularized optimal affine discriminant (ROAD)\citep{ROAD} and the direct sparse discriminant analysis (DSDA) \citep{Slda}. Denote the estimates of the discriminant directions given by ROAD and DSDA as $\hat\btheta^{\mathrm{ROAD}}$ and $\hat\btheta^{\mathrm{DSDA}}$, respectively. Then we have
\begin{eqnarray}
\hat\btheta^{\mathrm{ROAD}}(\lambda)&=&\arg\min_{\btheta} \btheta^\T\hat \bSigma\btheta+\lambda\Vert\btheta\Vert_1 \mbox{ s.t. $\btheta^\T(\hat\bmu_2-\hat\bmu_1)=1$}\label{ROAD}\\
\hat\btheta^{\mathrm{DSDA}}(\lambda)&=&\arg\min_{\btheta} \sum_i(Y^i-\theta_0-(\bX^i)^\T\btheta)^2+\lambda\Vert\btheta\Vert_1\label{DSDA}
\end{eqnarray}
We derive the following proposition to reveal the connections between our proposal $(K=2)$ and ROAD, DSDA. Note that the proofs of this proposition and all the subsequent lemmas and theorems can be found in the appendix.

\begin{proposition}\label{equiv}
Define $c_0(\lambda)=\hat\btheta^{\mathrm{MSDA}}(\lambda)^\T(\hat\bmu_2-\hat\bmu_1), c_1(\lambda)=\hat\btheta^{\mathrm{DSDA}}(\lambda)^\T(\hat\bmu_2-\hat\bmu_1)$ and $a=\frac{n|c_1(\lambda)|}{|c_0(\lambda)|}$. Then we have
\begin{eqnarray}
\hat\btheta^{\mathrm{MSDA}}(\lambda)&=&c_0(\lambda)\hat\btheta^{\mathrm{ROAD}}(\lambda/|c_0(\lambda)|),\label{prop1:eq1}\\
\hat\btheta^{\mathrm{MSDA}}(\lambda)&=&\dfrac{c_0(\lambda)}{c_1(a\lambda)}\hat\btheta^{\mathrm{DSDA}}(a\lambda).\label{prop1:eq2}
\end{eqnarray}
\end{proposition}

Proposition 1 shows that the classification direction by our proposal is identical to a classification direction by ROAD and a classification direction by DSDA. Consequently, 
our proposal $(K=2)$ has the same solution path as ROAD and DSDA.

\subsection{Algorithm}
Besides their solid theoretical foundation, LPD, ROAD and DSDA all enjoy computational efficiency. In particular, DSDA's computational complexity is the same as fitting a lasso linear regression model. In this section we show that our proposal for the multiclass problem can be solved by a very efficient algorithm. In light of this and Proposition 1, our proposal is regarded as the natural multiclass generalization of these successful binary sparse LDA methods.

We now present the efficient algorithm for solving (\ref{estimator2}).  For convenience write $\hat \bdelta^{k}=\hat \bmu_k-\hat \bmu_1$. Our algorithm is based on the following lemma.
\begin{lemma}\label{alg}
Given $\{\btheta_{.j'}, j' \neq j\}$, the solution of $\btheta_{.j}$ to \eqref{estimator2} is defined as
\begin{eqnarray}\label{bd1}
\arg\min_{\btheta_{.j}} \sum^K_{k=2}\frac{1}{2}(\theta_{kj}-\tilde \theta_{kj})^2+\frac{\lambda}{\hat \sigma_{jj}} \Vert \btheta_{.j}\Vert
\end{eqnarray}
where
$
\tilde \btheta_{k,j}=\frac{\hat \delta^k_{j}-\sum_{ l \neq j }\hat \sigma_{lj} \theta_{kl}}{\hat \sigma_{jj}}.
$
Let $\tilde \btheta_{.j}=(\tilde \theta_{2j},\ldots,\tilde \theta_{Kj})^{\T}$ and
$\Vert \tilde \btheta_{.j} \Vert=( \sum^K_{k=2} \tilde \theta^2_{kj} )^{1/2}$.
The solution to \eqref{bd1} is given by
\begin{eqnarray}\label{bd2}
\hat \btheta_{.j}= \tilde \btheta_{.j} \left(1-\frac{\lambda} {\Vert \tilde \btheta_{.j} \Vert} \right)_{+}.
\end{eqnarray}
\end{lemma}

Based on Lemma~\ref{alg} we use the following blockwise-descent algorithm to implement our multiclass sparse discriminant analysis.

\begin{algorithm}[Multiclass sparse discriminant analysis for a given penalization parameter]
\
\begin{enumerate}
\item Compute $\hat\bSigma$ and $\hat\bdelta^k$, $k=1,2,\ldots,K$;
\item Initialize $\hat\btheta_k^{(0)}$ and compute $\tilde\btheta_k^{(0)}$ accordingly;
\item For $m=1,\ldots, $ do the following loop until convergence:
\
for $j=1,\ldots, p$,
\begin{enumerate}
\item compute
\begin{eqnarray*}
\hat \btheta_{.j}^{(m)}&=& \tilde \btheta_{.j}^{(m-1)} \left(1-\frac{\lambda}{\Vert \tilde \btheta_{.j}^{(m-1)} \Vert} \right)_{+};
\eean
\item
update
\bean
\tilde \theta_{kj}&=&\frac{\hat \delta^k_{j}-\sum_{ l \neq j }\hat \sigma_{lj} \hat\theta_{kl}^{(m)}}{\hat \sigma_{jj}}.
\end{eqnarray*}
\end{enumerate}
\item Let $\hat\btheta_k$ be the solution at convergence. The output classifier is the usual linear discriminant classifier on $(\bX^\T\hat\btheta_2,\ldots,\bX^\T\hat\btheta_K)$.
\end{enumerate}
\end{algorithm}

 We have implemented our method in an R package {\tt msda} which is available on CRAN. Our package also handles the version of \eqref{estimator2} using an adaptive group lasso penalty, because both Lemma 1 and Algorithm 1 can be easily generalized to handle the adaptive group lasso penalty.

\section{Theory}
In this section we study theoretical properties of our proposal under the setting where $p$ can be much larger than $n$. Under regularity conditions we show that our method can consistently select the true subset of variables and at the same time consistently estimate the Bayes rule.

We begin with some useful notation. For a vector $\balpha$, $\vinf{\balpha}=\max_j|\alpha_j|,\vone{\balpha}=\sum_j|\alpha_j|$, while, for a matrix $\bOmega\in\mathbb{R}^{m\times n}$, $\vinf{\bOmega}=\max_{i}\sum_{j}|\omega_{ij}|$, $\vone{\bOmega}=\max_{j}\sum_{i}|\omega_{ij}|$. Define
\bean
&&\varphi=\max\{\vinf{\bSigma_{\cal D^C, D}},\vinf{\bSigma^{-1}_{\cal D, \cal D}}\}, \Delta=\max\{\vone{\bmu},\vone{\btheta}\};\\
&&\theta_{\min}=\min_{(k,j):\theta_{kj}\ne 0}|\theta_{kj}|, \theta_{\max}=\max_{(k,j)}|\theta_{kj}|;\\
&&
\Vert \bSigma_{\cal D^C,\cal D}\bSigma^{-1}_{\cal D,\cal D} \Vert_{\infty}=\eta^*.
\eean
Let $d$ be the cardinality of ${\cal D}$.

Define $\bt_{\cal D}\in \mathbb{R}^{d\times (K-1)}$ as the subgradient of the group lasso penalty at the true $\btheta_{\cal D}$ and
we assume the following condition:
\begin{enumerate}
\item[(C0)]
$
\max_{j \in {\cal D}^c} \{\sum_{k=2}^K (\bSigma_{j,\cal D}\bSigma_{\cal D,\cal D}^{-1}\bt_{k,\cal D})^2\}^{1/2}=\kappa <1.
$
\end{enumerate}
Condition (C0) is required to guarantee the selection consistency.
A condition similar to condition (C0) has been used to study the group lasso penalized regression model \citep{bach08}.

We further let $\varphi,\Delta, \eta^*,\kappa$ be fixed and assume the following regularity conditions:
\begin{enumerate}
\item[(C1)] There exists $c_1,C_1$ such that $\dfrac{c_1}{K}\le\pi_k\le\dfrac{C_1}{K}$ for $k=1,\ldots, K$ and $\dfrac{\theta_{\max}}{\theta_{\min}}<C_1$.
\item[(C2)] $n,p, \rightarrow \infty$ and $\dfrac{d^2\log{(pd)}}{n}\rightarrow 0$;
\item[(C3)] $\theta_{\min}\gg \{\dfrac{d^2\log{(pd)}}{n}\}^{1/2}$;
\item[(C4)] $\min_{k,k'}\{(\btheta_k-\btheta_{k'})^\T\bSigma(\btheta_k-\btheta_{k'})\}^{1/2}$ is bounded away from 0.
\end{enumerate}

Condition (C1) guarantees that we will have a decent sample size for each class.
Condition (C2) requires that $p$ cannot grow too fast with respect to $n$. This condition is very mild, because it can allow $p$ to grow at a nonpolynomial rate of $n$. In particular, if $d=O(n^{1/2-\alpha})$, then condition (C2) is satisfied if $\log{p}=o(n^{2\alpha})$. Condition (C3) guarantees that the nonzero coefficients are bounded away from 0, which is a common assumption in the literature. The lower bound of $\theta_{\min}$ tends to 0 under condition (C3). Condition (C4) is required such that all the classes can be separated from each other. If condition (C4) is violated, even the Bayes rule cannot work well.

In the following theorems, we let $C$ denote a generic positive constant that can vary from place to place. 

\begin{theorem}\label{asy}
\begin{enumerate}
\item Under conditions (C0)--(C1), there exists a generic constant $M$ such that, if $\lambda<\min\{\dfrac{\theta_{\min}}{8\varphi},M(1-\kappa)\}$, then with a probability greater than 
\begin{equation}
1-Cpd\exp(-Cn\dfrac{\epsilon^2}{Kd^2})-CK\exp(-C\dfrac{n}{K^2})-Cp(K-1)\exp(-Cn\dfrac{\epsilon^2}{K})
\end{equation} 

we have that $\hat{\cal D}=\cal D$, and $\vinf{\hat\btheta_{k}-\btheta_{k}^{\mathrm{Bayes}}}\le 4\varphi \lambda$ for $k=2,\ldots, K$.

\item If we further assume conditions (C2)--(C3),
we have that
if $
\{\dfrac{d^2\log{(pd)}}{n}\}^{1/2}
\ll \lambda \ll\theta_{\min}$,
 then with probability tending to 1, we have $\hat{\cal D}=\cal D$, and $\vinf{\hat\btheta_{k}-\btheta_{k}^{\mathrm{Bayes}}}\le 4\varphi \lambda$ for $k=2,\ldots, K$.
 
\end{enumerate}
\end{theorem}

Next, we show that our proposal is a consistent estimator of the Bayes rule in terms of the misclassification error rate. Define
\bean
R_n&=&\Pr(\hat Y(\hat\btheta_k,\hat\pi_k,k=1,\ldots, K)\ne Y\mid  \mathrm{observed \ data }),
\eean
where $\hat Y(\hat \btheta_k,\hat \pi_k,k=1,\ldots, K)$ is the prediction by our method. Also define $R$ as the Bayes error. Then we have the following conclusions.

\begin{theorem}\label{error}

\begin{enumerate}

\item Under conditions (C0)--(C1), there exists a generic constant $M_1$ such that, if $\lambda<\min\{\dfrac{\theta_{\min}}{8\varphi},M_1(1-\kappa)\}$, then with a probability greater than 
\begin{equation}
1-Cpd\exp(-Cn\dfrac{\epsilon^2}{Kd^2})-CK\exp(-C\dfrac{n}{K^2})-Cp(K-1)\exp(-Cn\dfrac{\epsilon^2}{K})
\end{equation} 

we have
\begin{equation}
|R_n-R|\le M_1\lambda^{1/3},
\end{equation}
for some generic constant $M_1$.

\item Under conditions (C0)--(C4), if $\lambda\rightarrow 0$, then with probability tending to 1, we have
$$
R_n\rightarrow R.
$$
\end{enumerate}

\end{theorem}

{\bf Remark 2}. Based on our proof we can further derive the asymptotic results by letting $K$ (the number of classes) diverge with $n$ to infinity.  We only need to use more cumbersome notion and bounds, but the analysis remains pretty much the same. To show a clearer picture of the theory, we have focused on the fixed $K$ case.

\section{Numerical Studies}
\subsection{Simulations}
We demonstrate our proposal by simulation. For comparison, we include the sparse optimal scoring and $\ell_1$ penalized Fisher's discriminant analysis in the simulation study. Four simulation models are considered where the dimension $p=800$ and the training set has a sample size $n=75K$, where $K$ is the number of classes in each model. We generate a validation set of size $n$ to select the tuning parameters and a testing set of size 1000 for each method. Recall that $\bbeta_k=\bSigma^{-1}\bmu_k$. We specify $\bbeta_k$ and $\bSigma$ as in the following four models and then let $\bmu_k=\bSigma\bbeta_k$. For simplicity, we say that a matrix $\bSigma$ has the AR($\rho$) structure if $\sigma_{jk}=\rho^{|j-k|}$ for $j,k=1,\ldots, p$; on the other hand, $\bSigma$ has the CS($\rho$) structure if $\sigma_{jk}=\rho$ for any $j\ne k$ and $\sigma_{jj}=1$ for $j=1,\ldots,p$.

Model 1: %sim7&sim23

$K=4$, $\beta_{jk}=1.6$ for $j=2k-1,2k; k=1,\ldots, K$ and $\beta_{jk}=0$ otherwise. The covariance matrix
 $\bSigma$ has the AR($0.5$) structure.

Model 2: %sim4&sim24

$K=6$, $\beta_{jk}=2.5$ for $j=2k-1,2k; k=1,\ldots, K$ and $\beta_{jk}=0$ otherwise. The covariance matrix $\bSigma=\bI_5\otimes \bOmega$, where $\bOmega$ has the CS($0.5$) structure. 

Model 3: %sim16&sim25

$K=4$, $\beta_{jk}=k+u_{jk}$ for $j=1,\ldots,K$, where $u_{jk}$ follows the uniform distribution over the interval $[-1/4,1/4]$; $\beta_{jk}=0$ otherwise. The covariance matrix $\bSigma$ has the CS($0.5$) structure.

Model 4: %sim21

$K=4$, $\beta_{jk}=k+u_{jk}$ for $j=1,\ldots,4$, where $u_{jk}$ follows the uniform distribution over the interval $[-1/4,1/4]$; $\beta_{jk}=0$ otherwise. The covariance matrix $\bSigma$ has the CS($0.8$) structure.

Model 5: %sim22

$K=4$, $\beta_{2,1}=\ldots=\beta_{2,8}=1.2$, $\beta_{3,1}=\ldots=\beta_{3,4}=-1.2$,
$\beta_{3,5}=\ldots=\beta_{3,8}=1.2$, $\beta_{4,2j-1}=-1.2, \beta_{4,2j}=1.2$ for $j=1,\ldots,4$; $\beta_{jk}=0$ otherwise. The covariance matrix $\bSigma$ has the AR($0.5$) structure.

Model 6: %sim10&sim26

$K=4$, $\beta_{2,1}=\ldots=\beta_{2,8}=1.2$, $\beta_{3,1}=\ldots=\beta_{3,4}=-1.2$
$\beta_{3,5}=\ldots=\beta_{3,8}=1.2$, $\beta_{4,2j-1}=-1.2, \beta_{4,2j}=1.2$ for $j=1,\ldots,4$; $\beta_{jk}=0$ otherwise. The covariance matrix $\bSigma$ has the AR($0.8$) structure. 

 The error rates of these methods are listed in Table~\ref{sim}. To compare variable selection performance, we report the number of correctly selected variables (C) and the number of incorrectly selected variables (IC) by each method. We want to highlight two observations from Table 1. First, our method is the best across all six models. Second, our method is a very good approximation of the Bayes rule in terms of both sparsity and misclassification error rate. Although our method tends to select a few more variables besides the true ones, this can be improved by using the adaptive group lasso penalty \citep{bach08}. Because the other two methods do not use the adaptive lasso penalty, we do not include the results of our method using the adaptive group lasso penalty for a fair comparison.

\renewcommand{\baselinestretch}{1}
\begin{table}
\begin{tabular}{c|cccc|cccc}
\hline
&Bayes&Our&Witten&Clemmensen&Bayes&Our&Witten&Clemmensen\\
\hline
&\multicolumn{4}{c}{Model 1}&\multicolumn{4}{|c}{Model 2}\\
Error(\%)&11.0&12.4&15.5&13&13.3&15.2&31.7&17\\
&(0.06)&(0.07)&(0.07)&(0.06)&(0.05)&(0.07)&(0.20)&(0.08)\\
C&8&8&8&8&12&12&12&12\\
&&(0)&(0)&(0)&&(0)&(0)&(0)\\
IC&0&10&126&5&0&15&19.5&16\\
&&(0.6)&(4.9)&(0.4)&&(0.7)&(1.5)&(0.3)\\
\hline
&\multicolumn{4}{c}{Model 3}&\multicolumn{4}{|c}{Model 4}\\
Error(\%)&8.8&9.4&14.1&12.7&5.3&5.7&7&7.6\\
&(0.06)&(0.09)&(0.06)&(0.08)&(0.06)&(0.08)&(0.05)&(0.07)\\
C&4&4&4&4&4&4&4&4\\
&&(0)&(0)&(0)&&(0)&(0)&(0)\\
IC&0&3&796&30&0&4&796&30\\
&&(0.4)&(0)&(0.2)&&(0.5)&(0)&(2.2)\\
\hline
&\multicolumn{4}{c}{Model 5}&\multicolumn{4}{|c}{Model 6}\\
Error(\%)&8.3&9.5&17.9&13.6&14.2&17.4&23.4&24.8\\
&(0.05)&(0.07)&(0.14)&(0.09)&(0.06)&(0.08)&(0.09)&(0.09)\\
C&8&8&8&8&8&8&8&6\\
&&(0)&(0)&(0)&&(0.0)&(0)&(0.1)\\
IC&0&6&97&4&0&0&4&3\\
&&(0.9)&(2.8)&(0.5)&&(0)&(0.5)&(0.3)\\
\hline
\end{tabular}
\caption{Simulation results for Models 1--6. The two competing methods are denoted by the first author of the original papers. In particular, Witten's method is the $\ell_1$ penalized Fisher's discriminant analysis, and Clemmensen's method is the sparse optimal scoring method. The reported numbers are medians based on 500 replicates. Standard errors are in parentheses. The quantity C is the number of correctly selected variables, and IC is the number of incorrectly selected variables.}\label{sim}
\end{table}
\renewcommand{\baselinestretch}{2}

\subsection{A real data example}
We further demonstrate the application of our method on the IBD dataset \citep{IBD}. This dataset contains 22283 gene expression levels from 127 people. These 127 people are either normal people, people with Crohn's disease or people with ulcerative colitis. This dataset can be downloaded from Gene Expression Omnibus with accession number GDS1615.  We randomly split the datasets with a 2:1 ratio in a balanced manner to form the training set and the testing set.

It is known that the marginal $t$-test screening \citep{FF2008} can greatly speed up the computation for linear discriminant analysis in binary problems. For a multiclass problem the natural generalization of $t$-test screening is the $F$-test screening. Compute the $F$-test statistic for each $X_j$ defined as
\begin{equation*}
f_j=\dfrac{\sum_{k=1}^K n_k(\hat \mu_{kj}-\hat{\bar\mu}_j)^2/(G-1)}{\sum_{i=1}^n (X_j^i-\hat \mu_{Y^i,j})^2/(n-G)},
\end{equation*}
where $\hat{\bar\mu}_j$ is the sample grand mean for $X_j$ and $n_g$ is the within-group sample size.
Based on the $F$-test statistic, we define the $F$-test screening by only keeping the predictors with $F$-test statistics among the $d_n$th largest. As recommended by many researchers \citep{FF2008, FS2010, Kfilter}, $d_n$ can be the same as the sample size, if we believe that the number of truly important variables is much smaller than the sample size. Therefore, we let $d_n=127$ for the current dataset.

We estimate the rules given by sparse optimal scoring, $\ell_1$ penalized Fisher's discriminant analysis and our proposal on the training set. The tuning parameters are chosen by 5 fold cross validation. Then we evaluate the classification errors on the testing set. The results based on 100 replicates are listed in Table~2. It can be seen that our proposal achieves the highest accuracy with the sparsest classification rule. This again shows that our method is a very competitive classifier.

\begin{table}
\begin{center}
\begin{tabular}{l|ccc}
\hline
&Our&Witten&Clemmensen\\
\hline
Error(\%)&7.32(0.972)&21.95(1.10)&9.76(0.622)\\
Fitted Model Size&25(0.7)&127(0)&27(0.5)\\
\hline
\end{tabular}
\caption{Classification and variable selection results on the real dataset. The two competing methods are denoted by the first author of the original papers. In particular, Witten's method is the $\ell_1$ penalized Fisher's discriminant analysis, and Clemmensen's method is the sparse optimal scoring method. All numbers are medians based on 100 random splits. Standard errors are in parentheses.}
\end{center}
\end{table}\label{realdata}

\section{Summary}
In this paper we have proposed a new formulation to derive sparse multiclass discriminant classifiers. We have shown that our proposal has a solid theoretical foundation and can be solved by a very efficient computational algorithm. Our proposal actually gives a unified treatment of the multiclass and binary classification problems. We have shown that the solution path of the binary version of our proposal is equivalent to that by ROAD and DSDA. Moreover, LPD is identical to the Dantzig selector formulation of our proposal for the binary case. In light of this evidence, our proposal is regarded as the natural multiclass generalization of those proven successful binary sparse LDA methods.

\section*{Appendices}

\subsection*{A.1 Connections with Fisher's discriminant analysis}
For simplicity, in this subsection we denote $\bmeta$ as the discriminant directions defined by Fisher's discriminant analysis in \eqref{fisher}, and $\btheta$ as the discriminant directions defined by Bayes rule. Our method gives a sparse estimate of $\btheta$. In this section, we discuss the connection between $\btheta$ and $\bmeta$, and hence the connection between our method and Fisher's discriminant analysis. We first comment on the advantage of directly estimating $\btheta$ rather than estimating $\bmeta$. Then we discuss how to estimate $\bmeta$ once $\hat\btheta$ is available. 

There are two advantages of estimating $\btheta$ rather than $\bmeta$. Firstly, estimating $\btheta$ allows for simultaneous estimation of all the discriminant directions. Note that \eqref{fisher} requires that $\bmeta_k^\T\bSigma\bmeta_l=0$ for any $l<k$. This requirement almost necessarily leads to a sequential optimization problem, which is indeed the case for sparse optimal scoring and $\ell_1$ penalized Fisher's discriminant analysis. In our proposal, the discriminant direction $\btheta_k$ is determined by the covariance matrix and the mean vectors $\bmu_k$ within Class k, but is not related to $\btheta_l$ for any $l\ne k$. Hence, our proposal can simultaneously estimate all the directions by solving a convex problem.  Secondly, it is easy to study the theoretical properties if we focus on $\btheta$. On the population level, $\btheta$ can be written out in explicit forms and hence it is easy to calculate the difference between $\btheta$ and $\hat\btheta$ in the theoretical studies. Since $\bmeta$ do not have closed-form solutions even when we know all the parameters, it is relatively harder to study its theoretical properties.

Moreover, if one is specifically interested in the discriminant directions $\bmeta$, it is very easy to obtain a sparse estimate of them once we have a sparse estimate of $\btheta$. For convenience, for any positive integer $m$, denote $0_m$ as an $m$-dimensional vector with all entries being 0, $1_m$ as an $m$-dimensional vector with all entries being 1, and $\bI_m$ as the $m\times m$ identity matrix. The following lemma provides an approach to estimating $\bmeta$ once $\hat\btheta$ is available. The proof is relegated to Section A.2.

\begin{lemma}\label{lemma:fisher}

The discriminant directions $\bmeta$ contain all the right eigenvectors of $\btheta_0\bPi\bdelta_0^\T$ corresponding to positive eigenvalues,
where $\btheta_0=(0_p,\btheta)$, $\bPi=\bI_K-\frac{1}{K}\mathrm{1}_{K}\mathrm{1}_{K}^\T$, and $\bdelta_0=(\bmu_{1}-\bar\bmu,\ldots,\bmu_{K}-\bar\bmu)$ with $\bar\bmu=\sum_{k=1}^K \pi_k\bmu_k$.
\end{lemma}

Therefore, once we have obtained a sparse estimate of $\btheta$, we can estimate $\bmeta$ as follows. Without loss of generality write $\hat\btheta=(\hat\btheta_{\hat{\cal D}}^\T,0)^\T$, where $\hat{\cal D}=\{j:\hat\btheta_{\cdot j}\ne 0\}$. Then $\hat\btheta_0=(0,\hat\btheta)$. On the other hand, set $\hat\bdelta_0=(\hat\bmu_1-\hat{\bar\bmu},\ldots,\hat\bmu_K-\hat{\bar\bmu})$ where $\hat\bmu_k$ are sample estimates and $\hat{\bar{\bmu}}=\sum_{k=1}^K\hat{\pi}_k\hat\bmu_k$. It follows that $\hat\btheta_0\bPi\hat\bdelta_{0}=((\hat\btheta_{0,\hat{\cal D}}\bPi\hat\bdelta_{0,\hat{\cal D}}^\T)^\T,0)^\T$. Consequently, we can perform eigen-decomposition on $\hat\btheta_{0,\hat{\cal D}}\bPi\hat\bdelta_{0,\hat{\cal D}}^\T$ to obtain $\hat \bmeta_{\hat{\cal D}}$. Because $\hat{\cal D}$ is a small subset of the original dataset, this decomposition will be computationally efficient. Then $\hat\bmeta$ would be $(\hat \bmeta_{\hat{\cal D}}^\T,0)^\T$.

\subsection*{A.2 Technical Proofs}

\begin{proof}[Proof of Proposition~\ref{equiv}]
We first show \eqref{prop1:eq1}.

For a vector $\btheta\in\mathbb{R}^p$, Define
\begin{eqnarray}
L^{\mathrm{MSDA}}(\btheta,\lambda)&=&\frac{1}{2}\btheta^\T \hat \bSigma \btheta-(\hat \bmu_2-\hat \bmu_1)^{\T} \btheta+\lambda \Vert\btheta\Vert_1,\\
L^{\mathrm{ROAD}}(\btheta,\lambda)&=&\btheta^\T\hat\bSigma\btheta+\lambda \Vert\btheta\Vert_1
\end{eqnarray}

Set $\tilde \btheta=c_0(\lambda)^{-1}\hat \btheta^{\mathrm{MSDA}}(\lambda)$. Since $\tilde\btheta^\T(\hat\bmu_2-\bmu_1)=1$, it suffices to check that, for any $\tilde\btheta'$ such that $(\tilde\btheta')^\T(\hat\bmu_2-\bmu_1)=1$, we have $L^{\mathrm{ROAD}}(\tilde\btheta,\frac{\lambda}{|c_0(\lambda)|})\le L^{\mathrm{ROAD}}(\tilde\btheta',\frac{\lambda}{|c_0(\lambda)|})$. Now for any such $\tilde \btheta'$,
\begin{equation}
L^{\mathrm{MSDA}}(c_0(\lambda)\tilde\btheta',\lambda)=c_0(\lambda)^2 L^{\mathrm{ROAD}}(\tilde\btheta',\frac{\lambda}{|c_0(\lambda)|}) -c_0(\lambda)
\end{equation}
Similarly,
\begin{equation}
L^{\mathrm{MSDA}}(c_0(\lambda)\tilde\btheta,\lambda)=c_0(\lambda)^2 L^{\mathrm{ROAD}}(\tilde\btheta,\frac{\lambda}{|c_0(\lambda)|}) -c_0(\lambda).
\end{equation}
Since $L^{\mathrm{MSDA}}(c_0(\lambda)\tilde\btheta,\lambda)\le L^{\mathrm{MSDA}}(c_0(\lambda)\tilde\btheta',\lambda)$, we have \eqref{prop1:eq1}.

On the other hand, by Theorem 1 in \cite{equivalence}, we have
\begin{equation}
\hat\btheta^{\mathrm{DSDA}}(\lambda)=c_1(\lambda)\hat\btheta^{\mathrm{ROAD}}(\frac{\lambda}{n|c_1(\lambda)|})
\end{equation}
Therefore,
\begin{eqnarray}
\hat\btheta^{\mathrm{ROAD}}(\frac{\lambda}{|c_0(\lambda)|})&=&\hat\btheta^{\mathrm{ROAD}}\left((\frac{n|c_1(\lambda)|\lambda}{|c_0(\lambda)|})/(n|c_1(\lambda)|)\right)\\
&=&\left(c_1(\frac{n|c_1(\lambda)|\lambda}{|c_0(\lambda)|})\right)^{-1}\hat\btheta^{\mathrm{DSDA}}\left(\frac{n|c_1(\lambda)|\lambda}{|c_0(\lambda)|}\right)\\
&=&(c_1(a\lambda))^{-1}\hat\btheta^{\mathrm{DSDA}}(a\lambda)\label{prop1:eq3}
\end{eqnarray}
Combine \eqref{prop1:eq3} with \eqref{prop1:eq1} and we have \eqref{prop1:eq2}.
\end{proof}

\begin{proof}[Proof of Lemma~\ref{alg}]
We start with simplifying the first part of our objective function, $\frac{1}{2}\btheta^\T_k \hat \bSigma \btheta_k-(\hat \bmu_k-\hat \bmu_1)^{\T} \btheta_k$. 

First, note that
\begin{eqnarray}
&&\frac{1}{2}\btheta^\T_k \hat \bSigma \btheta_k=\frac{1}{2}\sum_{l,m=1}^p\theta_{kl}\theta_{km}\hat\sigma_{lm}\\
&=&\frac{1}{2}\theta_{kj}^2\hat\sigma_{jj}+\frac{1}{2}\sum_{l\ne j}\theta_{kl}\theta_{kj}\hat\sigma_{lj}+\frac{1}{2}\sum_{m\ne j}\theta_{kj}\theta_{km}\hat\sigma_{jm}+\frac{1}{2}\sum_{l\ne j, m\ne j}\theta_{kl}\theta_{km}\hat\sigma_{lm}\\
\end{eqnarray}

Because $\hat\sigma_{lj}=\hat\sigma_{jl}$, we have $\sum_{l\ne j}\theta_{kl}\theta_{kj}\hat\sigma_{lj}=\sum_{m\ne j}\theta_{kj}\theta_{km}\hat\sigma_{jm}$. It follows that
\begin{eqnarray}
\frac{1}{2}\btheta^\T_k \hat \bSigma \btheta_k&=&\frac{1}{2}\theta_{kj}^2\hat\sigma_{jj}+\sum_{l\ne j}\theta_{kj}\theta_{kl}\hat\sigma_{lj}+\frac{1}{2}\sum_{l\ne j, m\ne j}\theta_{kl}\theta_{km}\hat\sigma_{lm}\label{eq1:lemma1}
\end{eqnarray}

Then recall that $\hat \bdelta^k=\hat\bmu_k-\hat\bmu_1$. We have 
\begin{equation}\label{eq2:lemma1}
(\hat \bmu_k-\hat \bmu_1)^{\T} \btheta_k=\sum_{l=1}^p\delta^k_l\theta_{kl}=\delta^k_j\theta_{kj}+\sum_{l\ne j}\delta^k_l\theta_{kl}
\end{equation}

Combine \eqref{eq1:lemma1} and \eqref{eq2:lemma1} and we have
\begin{eqnarray}
&&\frac{1}{2}\btheta^\T_k \hat \bSigma \btheta_k-(\hat \bmu_k-\hat \bmu_1)^{\T} \btheta_k\\
&=&\frac{1}{2}\theta_{kj}^2\hat\sigma_{jj}+\sum_{l\ne j}\theta_{kj}\theta_{kl}\hat\sigma_{lj}+\frac{1}{2}\sum_{l\ne j, m\ne j}\theta_{kl}\theta_{km}\hat\sigma_{lm}-\delta^k_j\theta_{kj}-\sum_{l\ne j}\delta^k_l\theta_{kl}\\
&=&\frac{1}{2}\theta^2_{kj}\hat \sigma_{jj}+(\sum_{ l \neq j }\hat \sigma_{l,j} \theta_{kl}-
\hat \delta^{k}_{j})\theta_{kj}+\frac{1}{2} \sum_{m \neq j, l \neq j} \theta_{kl}\theta_{km}\hat \sigma_{lm}- \sum_{l \neq j}\hat \delta^{k}_{l}\theta_{kl}
\end{eqnarray}

Note that the last two terms does not involve $\btheta_{.j}$. Therefore, given $\{\btheta_{.j'}, j' \neq j\}$, the solution of $\btheta_{.j}$ is defined as
\begin{eqnarray*}
\arg\min_{\btheta_{2,j},\ldots,\btheta_{K,j}} \sum^{K}_{k=2} \{
\frac{1}{2}\theta^2_{kj}\hat \sigma_{jj}+(\sum_{ l \neq j }\hat \sigma_{lj} \theta_{kl}-
\hat \delta^{k}_{j})\theta_{kj}\}+\lambda  \Vert \btheta_{.j}\Vert,
\end{eqnarray*}
which is equivalent to \eqref{bd1}. It is easy to get \eqref{bd2} from \eqref{bd1} \citep{groupedlasso}.
\end{proof}

In what follows we use $C$ to denote a generic constant for convenience.

Now we define an oracle ``estimator" that relies on the knowledge of $\cal D$ for a specific tuning parameter $\lambda$:
\beq\label{formula:oracle}
\hat\btheta_{\cal D}^{\mathrm{oracle}}=\arg\min_{\btheta_{2,\cal D},\ldots,\btheta_{K,\cal D}} \sum^{K}_{k=2}\{\frac{1}{2}\btheta_{k,\cal D}^\T \hat \bSigma_{\cal D, \cal D} \btheta_{k,\cal D}-(\hat \bmu_{k,\cal D}-\hat \bmu_{1,\cal D})^{\T} \btheta_{k,\cal D}\}+\lambda\sum_{j\in\cal D}  \Vert \theta_{.j}\Vert .
\eeq

The proof of Theorem~\ref{asy} is based on a series of technical lemmas.

\begin{lemma}\label{KKT}
Define $\hat\btheta_{\cal D}^{\mathrm{oracle}}(\lambda)$ as in \eqref{formula:oracle}. Then $\hat\btheta_k=(\hat\btheta_{k,\cal D}^{\mathrm{oracle}},0),k=2,\ldots, K$ is the solution to \eqref{estimator2} if
\bea\label{KKT1}
\max_{j \in {\cal D}^c}[\sum_{k=2}^K \{(\hat\bSigma_{\cal D^C, \cal D}\hat \btheta_{k,\cal D}^{\mathrm{(oracle)}})_j-(\hat \mu_{kj}-\hat\mu_{1j})\}^2]^{1/2}<\lambda.
\eea
\end{lemma}

\begin{proof}[Proof of Lemma~\ref{KKT}]
The proof is completed by checking that $\hat\btheta_k=(\hat\btheta_{k,\cal D}^{\mathrm{oracle}}(\lambda),0)$ satisfies the KKT condition of \eqref{estimator2}.
\end{proof}

\begin{lemma}\label{muD}
For each $k$,
$
\bSigma_{\cal D^C,\cal D}\bSigma^{-1}_{\cal D,\cal D}(\bmu_{k,\cal D}-\bmu_{1,\cal D})=\bmu_{k,\cal D^C}-\bmu_{1,\cal D^C}.
$
\end{lemma}
\begin{proof}[Proof of Lemma~\ref{muD}]
For each $k$, we have $\btheta_{k,\cal D^C}=0$. By definition, $\btheta_{\cal D^C}=(\bSigma^{-1}(\bmu_k-\bmu_1))_{\cal D^C}$. Then by block inversion, we have that
\beqn
\btheta_{k,\cal D^C}=-(\bSigma_{\cal D^C,\cal D^C}-\bSigma_{\cal D^C,\cal D}\bSigma_{\cal D,\cal D}\bSigma_{\cal D,\cal D^C})^{-1}(\bSigma_{\cal D^C,\cal D}\bSigma_{\cal D,\cal D}^{-1}(\bmu_{k,\cal D}-\bmu_{1,\cal D})-(\bmu_{k,\cal D^C}-\bmu_{1,\cal D^C})),
\eeqn
and the conclusion follows.
\end{proof}

\begin{proposition}\label{concentration}
There exist a constant
$\epsilon_0 $ such that for any $\epsilon \leq
\epsilon_0$ we have
\begin{eqnarray}
&& \pr\{|(\hat \mu_{kj}-\hat
\mu_{1j})-(\mu_{kj}-\mu_{1j})| \ge \epsilon\} \le C\exp(-C\dfrac{n\epsilon^2}{K})+C\exp(-\dfrac{Cn}{K^2}),\label{lamma1.eq2}\\\
&& k=2,\ldots, K, \ j=1,\ldots,p; \nonumber \\
&&\pr(|\hat \sigma_{ij}-\sigma_{ij}| \ge \epsilon) \le 2\exp(-C\dfrac{n\epsilon^2}{K})+2\exp(-\dfrac{Cn}{K^2}),\ i,j=1,\ldots,p. \label{lamma1.eq1}
\eea
\end{proposition}

\begin{proof}[Proof of Proposition~\ref{concentration}]
We first show \eqref{lamma1.eq2}. Note that, by Chernoff bound
\bean
\pr(|\hat\mu_{kj}-\mu_{kj}|\ge\epsilon)&\le&E(\pr(|\hat\mu_{kj}-\mu_{kj}|\ge\epsilon\mid Y))\le E(C\exp(-Cn_k\epsilon^2))\\
&\le&2\exp(-C\dfrac{n\epsilon^2}{K})+2\exp(-\dfrac{Cn}{K^2}).
\eean
A similar inequality holds for $\hat\mu_{1j}$, and \eqref{lamma1.eq2} follows.

For \eqref{lamma1.eq1}, note that
\bean
\hat \sigma_{ij}&=&\dfrac{1}{n-K}\sum_{k=1}^K\sum_{Y^m=k}(X^m_i-\hat\mu_{ki})(X^m_j-\hat\mu_{kj})\\
&=&\dfrac{1}{n-K}\sum_{k=1}^K\sum_{Y^m=k}(X^m_i-\mu_i^m)(X^m_j-\mu_j^m)+\dfrac{1}{n-K}\sum_{k=1}^K n_k(\hat\mu_{ki}-\mu_{ki})(\hat\mu_{kj}-\mu_{kj}) \quad \\
&=&\hat\sigma_{ij}^{(0)}+\dfrac{1}{n-K}\sum_{k=1}^K n_k(\hat\mu_{ki}-\mu_{ki})(\hat\mu_{kj}-\mu_{kj}).
\eean
Now by Chernoff bound,
$
\pr(|\hat\sigma_{ij}^{(0)}-\sigma_{ij}|\ge\epsilon)\le C\exp(-Cn\epsilon^2).
$
Combining this fact with \eqref{lamma1.eq2}, we have the desired result.

\end{proof}

Now we consider two events depending on a small $\epsilon>0$:
\bean
A(\epsilon)&=&\{|\hat \sigma_{ij}-\sigma_{ij}|<\dfrac{\epsilon}{d}\mbox{ for any $i=1,\cdots,p$ and $j\in\cal D$}\},\\
B(\epsilon)&=&\{|(\hat \mu_{kj}-\hat
\mu_{1j})-(\mu_{kj}-\mu_{1j})|< \epsilon\mbox{ for any $k$ and $j$}\}.
\eean

By simple union bounds, we can derive Lemma 4 and Lemma 5.
\begin{lemma}\label{prob.bounds}
There exist a constant
$\epsilon_0 $ such that for any $\epsilon \leq
\epsilon_0$ we have
\begin{enumerate}
\item $\pr(A(\epsilon))\ge 1-Cpd \exp(-Cn\dfrac{\epsilon^2}{Kd^2})-CK\exp(-\dfrac{Cn}{K^2})$;
\item $\pr(B(\epsilon))\ge 1-Cp(K-1)\exp(-C\dfrac{n\epsilon^2}{K})-CK\exp(-\dfrac{Cn}{K^2})$;
\item $\pr(A(\epsilon)\cap B(\epsilon)) \ge 1-\gamma(\epsilon)$,
where
\beqn\label{delta}
\gamma(\epsilon)=Cpd\exp(-C\dfrac{n\epsilon^2}{d^2})+Cp(K-1)\exp(-C\dfrac{n\epsilon^2}{K})+2CK\exp(-\dfrac{Cn}{K^2}).
\eeqn
\end{enumerate}
\end{lemma}

%Throughout the rest of the proof, we assume that both $A(\epsilon)$ and $B(\epsilon)$ have occurred for some carefully chosen %$\epsilon<\epsilon_0$. It is easy to derive the following lemma.

\begin{lemma}\label{lemma1}
Assume that both $A(\epsilon)$ and $B(\epsilon)$ have occurred. We have the following conclusions:
\begin{eqnarray*}
&&\Vert \hat\bSigma_{\cal D, \cal D}-\bSigma_{\cal D, \cal D} \Vert_{\infty}< \epsilon;\label{lamma1.eq3}\\
&&\Vert \hat\bSigma_{\cal D^C, \cal D}-\bSigma_{\cal D^C, \cal D} \Vert_{\infty} < \epsilon;\\
&&\Vert(\hat\bmu_{k}-\hat\bmu_{1})-(\bmu_k-\bmu_1) \Vert_{\infty}< \epsilon;\label{lamma1.eq5}\\
&&\Vert(\hat\bmu_{k,\cal D}-\hat\bmu_{1,\cal D})-(\bmu_{k,\cal D}-\bmu_{1,\cal D}) \Vert_{1}< \epsilon.\label{lamma1.eq6}
\end{eqnarray*}
\end{lemma}

\begin{lemma}\label{lemma2}
If both $A(\epsilon)$ and $B(\epsilon)$ have occurred for $\epsilon<\dfrac{1}{\varphi}$, we have
\bean\label{siginv}
&&\vone{\hat \bSigma_{\cal D,\cal D}^{-1}-\bSigma_{\cal D,\cal D}^{-1}}
< \epsilon \varphi^2(1-\varphi \epsilon)^{-1},\\
&&\Vert
\hat\bSigma_{\cal D^C,\cal D}(\hat \bSigma_{\cal D,\cal D})^{-1}-\bSigma_{\cal D^C,\cal D}(\bSigma_{\cal D,\cal D})^{-1}\Vert_{\infty}< \dfrac{\varphi\epsilon}{1-\varphi\epsilon}.
\eean

\end{lemma}

\begin{proof}[Proof of Lemma~\ref{lemma2} ]
 Let $\eta_1= \Vert
\hat\bSigma_{\cal D, \cal D}-\bSigma_{\cal D, \cal D}\Vert_{\infty}  $,
 $\eta_2= \Vert \hat\bSigma_{\cal D^C, \cal D}-\bSigma_{\cal D^C, \cal D}\Vert_{\infty}  $ and
$\eta_3= \Vert (\hat\bSigma_{\cal D, \cal D})^{-1}-(\bSigma_{\cal D, \cal D})^{-1} \Vert_{\infty}$.
First we have
\beqn
\eta_3  \le \Vert {(\hat\bSigma_{\cal D, \cal D})^{-1}}\Vert_{\infty} \times \Vert
(\hat\bSigma_{\cal D, \cal D}-\bSigma_{\cal D, \cal D}) \Vert_{\infty} \times \Vert
(\bSigma_{\cal D, \cal D})^{-1}\Vert_{\infty} = (\varphi+\eta_3) \varphi \eta_1.
\eeqn
On the other hand,
\begin{eqnarray*}\label{lemma2.eq1}
\Vert
\hat\bSigma_{\cal D^C, \cal D}(\hat\bSigma_{\cal D, \cal D})^{-1}-\bSigma_{\cal D^C, \cal D}(\bSigma_{\cal D, \cal D})^{-1}\Vert_{\infty}
 &\le&   \Vert \hat\bSigma_{\cal D^C, \cal D}-\bSigma_{\cal D^C, \cal D}\Vert_{\infty} \times \Vert
(\hat\bSigma_{\cal D, \cal D})^{-1}-(\bSigma_{\cal D, \cal D})^{-1}\Vert_{\infty}\nonumber\\
&&+\Vert \hat\bSigma_{\cal D^C, \cal D}-\bSigma_{\cal D^C, \cal D}\Vert_{\infty} \times
\Vert(\bSigma_{\cal D, \cal D})^{-1}\Vert_{\infty} \nonumber\\
&&+\vinf{\bSigma_{\cal D^C, \cal D}}\times\Vert
(\hat\bSigma_{\cal D, \cal D})^{-1}-(\bSigma_{\cal D, \cal D})^{-1}\Vert_{\infty} \\
&\le&\eta_2\eta_3+\eta_2\varphi+\varphi\eta_3.\nonumber
\end{eqnarray*}
By $\varphi \eta_1<1$ we have $\eta_3 \le \varphi^2
\eta_1 (1-\varphi \eta_1)^{-1}$ and hence
\begin{equation*}\label{lemma2.eq3}
\Vert
\hat\bSigma_{\cal D^C, \cal D}(\hat\bSigma_{\cal D, \cal D})^{-1}-\bSigma_{\cal D^C, \cal D}(\bSigma_{\cal D, \cal D})^{-1}\Vert_{\infty}<\dfrac{\varphi\epsilon}{1-\varphi\epsilon}.
\end{equation*}

\end{proof}

\begin{lemma}\label{uporacle}
Define
\beq\label{theta0}
\hat\btheta_{k,\cal D}^{0}=\hat\bSigma^{-1}_{\cal D, \cal D}(\hat\bmu_{k,\cal D}-\hat\bmu_{1,\cal D}).
\eeq
Then
$
\vone{\hat\btheta_{k,\cal D}^{0}-\btheta_{k,\cal D}}\le \dfrac{\varphi\epsilon(1+\varphi\Delta)}{1-\varphi\epsilon}.
$
\end{lemma}
\begin{proof}[Proof of Lemma~\ref{uporacle}]
By definition, we have
\bean
 &&\Vert\hat\bSigma_{\cal D,\cal D}^{-1}(\hat\bmu_{k,\cal D}-\hat\bmu_{1,\cal D})-\bSigma_{\cal D,\cal D}^{-1}(\bmu_{k,\cal D}-\bmu_{1,\cal D})\Vert_{1}\\
 &&\le \Vert\hat\bSigma_{\cal D,\cal D}^{-1}-\bSigma_{\cal D,\cal D}^{-1}\Vert_{1} \Vert(\hat\bmu_{k,\cal D}-\hat\bmu_{1,\cal D})-(\bmu_{k,\cal D}-\bmu_{1,\cal D})\Vert_{1}\\
&&+\Vert\bSigma_{\cal D,\cal D}^{-1}\Vert_{1}\Vert(\hat\bmu_{k,\cal D}-\hat\bmu_{1,\cal D})-(\bmu_{k,\cal D}-\bmu_{1,\cal D})\Vert_{1}+\Vert\hat\bSigma_{\cal D,\cal D}^{-1}-\bSigma_{\cal D,\cal D}^{-1}\Vert_{1} \Vert\bmu_{k,\cal D}-\bmu_{1,\cal D}\Vert_{1} \quad \\
&\le& \dfrac{\varphi\epsilon(1+\varphi\Delta)}{1-\varphi\epsilon}.
\eean

\end{proof}

\begin{lemma}\label{oracle}
If $A(\epsilon)$ and $B(\epsilon)$ have occurred for $\epsilon<\min\{\frac{1}{2\varphi},\dfrac{\lambda}{1+\varphi\Delta}\}$, then for all $k$
\beqn
\Vert\hat\btheta_{k,\cal D}^{(\mathrm{oracle})}(\lambda)-\btheta_{k,\cal D}\Vert_{\infty}\le 4\lambda\varphi.
\eeqn

\end{lemma}

\begin{proof}[Proof of Lemma~\ref{oracle}]
Observe
$
\hat\btheta_k^{\mathrm{oracle}}=\hat\bSigma_{\cal D,\cal D}^{-1}(\hat\bmu_{k,\cal D}-\hat\bmu_{1,\cal D})-\lambda \hat\bSigma_{\cal D,\cal D}^{-1}\hat \bt_{k,\cal D}.
$
Therefore,
\bean
&&\Vert\hat\btheta^{\mathrm{oracle}}_{k,\cal D}-\btheta_{k,\cal D}\Vert_{\infty} \nonumber\\
&\le & \Vert\hat\btheta^{0}_{k,\cal D}-\btheta_{k,\cal D}\Vert_{\infty}+\lambda\Vert\hat\bSigma_{\cal D,\cal D}^{-1}-\bSigma_{\cal D,\cal D}^{-1}\Vert_{1}\Vert\hat \bt_{k,\cal D}\Vert_{\infty}+\lambda\Vert\bSigma_{\cal D,\cal D}^{-1}\Vert_{1}\Vert\hat \bt_{k,\cal D}\Vert_{\infty}
\eean
where $\hat\btheta_{k,D}^0$ is defined as in \eqref{theta0}. Now $\Vert\hat \bt_{k,\cal D}\Vert_{\infty}\le 1$ and we have
\bean
\Vert\hat\btheta^{\mathrm{oracle}}_{k,\cal D}-\btheta_{k,\cal D}\Vert_{\infty}\le \dfrac{\varphi \epsilon(1+\varphi\Delta)+\lambda\varphi}{1-\varphi\epsilon}<4\varphi\lambda.
\eean

\end{proof}

\begin{lemma}\label{CS}
For a sets of real numbers $\{a_1,\ldots,a_N\}$, if $\sum_{i=1}^N a_i^2\le \kappa^2<1$, then $\sum_{i=1}^N (a_i+b)^2<1$ as long as $b<\dfrac{1-\kappa}{\sqrt{N}}$.
\end{lemma}
\begin{proof}
By the Cauchy-Schwartz inequality, we have that 
\bea
\sum_{i=1}^N (a_i+b)^2&=& \sum_{i=1}^N a_i^2+2\sum_{i=1}^N a_i b+Nb^2\\
&\le& \sum_{i=1}^N a_i^2+2\sqrt{(\sum_{i=1}^N a_i^2) \cdot Nb^2 }+Nb^2\\
&\le & \kappa^2+2\kappa \sqrt{Nb^2}+Nb^2
\eea
which is less than 1 when $b<\dfrac{1-\kappa}{\sqrt{N}}$.
\end{proof}

We are ready to complete the proof of Theorem~\ref{asy}.

\begin{proof}[Proof of Theorem~\ref{asy}]
We first consider the first conclusion. For any $\lambda<\frac{\theta_{\min}}{8\varphi}$ and $\epsilon<\min\{\frac{1}{2\varphi},\dfrac{\lambda}{1+\varphi\Delta}\}$,
consider the event $A(\epsilon) \cap B(\epsilon)$. By Lemmas \ref{KKT}, \ref{prob.bounds} \& \ref{oracle} it suffices to verify \eqref{KKT1}.

For any $j \in {\cal D}^c$, by Lemma~\ref{muD} we have
\bean
&&|(\hat\bSigma_{\cal D^C,\cal D}\hat\btheta^{\mathrm{(oracle)}}_{k,\cal D})_j-(\hat\mu_{kj}-\hat\mu_{1j})|\\
&\le& |(\hat\bSigma_{\cal D^C,\cal D}\hat\btheta^{\mathrm{(oracle)}}_{k,\cal D})_j-(\bSigma_{\cal D^C,\cal D}\btheta_{k,\cal D})_j|+|(\hat\mu_{kj}-\hat\mu_{1j})-(\mu_{kj}-\mu_{1j})|\\
&\le&|(\hat\bSigma_{\cal D^C,\cal D}\hat\btheta^{\mathrm{(oracle)}}_{k,\cal D})_j-(\bSigma_{\cal D^C,\cal D}\btheta_{k,\cal D})_j|+\epsilon\\
&\le&|(\hat\bSigma_{\cal D^C,\cal D}\hat\btheta^{\mathrm{(0)}}_{k,\cal D})_j-(\bSigma_{\cal D^C,\cal D}\btheta_{k,\cal D})_j|+\epsilon+\lambda|(\hat\bSigma_{\cal D^C,\cal D}\hat\bSigma_{\cal D, \cal D}^{-1}\hat \bt_{k,\cal D})_j|\\
%&=&L_{1j}+L_{2j}.
\eean
\bea
&&|(\hat\bSigma_{\cal D^C,\cal D}\hat\btheta^{\mathrm{(oracle)}}_{k,\cal D})_j-(\bSigma_{\cal D^C,\cal D}\btheta_{k,\cal D})_j|+\epsilon \nonumber \\
&\le &  \Vert (\hat\bSigma_{\cal D^C,\cal D})_j-(\bSigma_{\cal D^C,\cal D})_j \Vert_1   \vinf{\hat\btheta^{0}_{k,\cal D}-\btheta_{k,\cal D}}+\vinf{\btheta_{k,\cal D}}
\Vert (\hat\bSigma_{\cal D^C,\cal D})_j-(\bSigma_{\cal D^C,\cal D})_j \Vert_1 \nonumber  \\
&&+ \Vert (\bSigma_{\cal D^C,\cal D})_j \Vert_{1}  \vinf{\hat\btheta_{k,\cal D}^0-\btheta_{k,\cal D}}+\epsilon \nonumber  \\
&\le & C \epsilon. \label{jack}
\eea
\bean
&&|(\hat\bSigma_{\cal D^C, \cal D}\hat\bSigma_{\cal D, \cal D}^{-1}\hat\bt_{k,\cal D})_j-(\bSigma_{\cal D^C,\cal D}\bSigma_{\cal D, \cal D}^{-1}\bt_{k,\cal D})_j|\\
&\le&\vinf{\hat\bSigma_{\cal D^C, \cal D}\hat\bSigma_{\cal D, \cal D}^{-1}-\bSigma_{\cal D^C,\cal D}\bSigma_{\cal D, \cal D}^{-1}}\vinf{\hat\bt_{k,\cal D}-\bt_{k,\cal D}} \nonumber \\
&&+\vinf{\bSigma_{\cal D^C,\cal D}\bSigma_{\cal D, \cal D}^{-1}}\vinf{\hat\bt_{k,\cal D}-\bt_{k,\cal D}} +\vinf{\hat\bSigma_{\cal D^C, \cal D}\hat\bSigma_{\cal D, \cal D}^{-1}-\bSigma_{\cal D^C,\cal D}\bSigma_{\cal D, \cal D}^{-1}}   |(\bt_{k,\cal D})_j|\\
%&\le& \dfrac{C\varphi}{\theta_{\min}\surd{(K-1)}}\lambda+\dfrac{\varphi\epsilon}{1-\varphi\epsilon}.
\eean
\bean
|\hat t_{kj}-t_{kj}| &=&
%%=|\dfrac{\hat\theta_{kj}}{ \Vert \hat\theta_{.j} \Vert }-\dfrac{\theta_{kj}}{ \Vert \theta_{.j} \Vert }|\\
|\dfrac{\hat\theta_{kj}\Vert \theta_{.j} \Vert -\theta_{kj}\Vert \hat\theta_{.j} \Vert }{\Vert \theta_{.j} \Vert \Vert \hat\theta_{.j} \Vert }|\\
&\le&\dfrac{|\hat\theta_{kj}-\theta_{kj}|\Vert \theta_{.j} \Vert+\theta_{\max} \Vert \theta_{.j}- \hat\theta_{.j} \Vert }{  \Vert \theta_{.j} \Vert \Vert \hat\theta_{.j} \Vert   }\\
&\le&\dfrac{C\varphi}{\theta_{\min}\surd{(K-1)}}\lambda.
\eean
Therefore, 
\bea
&&\lambda|(\hat\bSigma_{\cal D^C,\cal D}\hat\bSigma_{\cal D, \cal D}^{-1}\hat \bt_{k,\cal D})_j| \nonumber \\
&\le &\lambda|(\bSigma_{\cal D^C,\cal D}\bSigma^{-1}_{\cal D,\cal D}\bt_{k,\cal D})_j|+\lambda(\dfrac{C\varphi\epsilon}{1-\varphi\epsilon}+\eta^*\dfrac{C\varphi\lambda}{\theta_{\min}\sqrt{K-1}})\\
&\le& \lambda|(\bSigma_{\cal D^C,\cal D}\bSigma^{-1}_{\cal D,\cal D}\bt_{k,\cal D})_j|+C\lambda^2 \label{tom}
\eea
Under condition (C0), it follows from (\ref{jack}) and (\ref{tom}) that
\beq
|(\hat\bSigma_{\cal D^C,\cal D}\hat\btheta^{\mathrm{(oracle)}}_{k,\cal D})_j-(\hat\mu_{kj}-\hat\mu_{1j})|\le \lambda|(\bSigma_{\cal D^C,\cal D}\bSigma^{-1}_{\cal D,\cal D}\bt_{k,\cal D})_j| +C\lambda^2
\eeq
Combine condition (C0) with Lemma~\ref{CS}, we have that, there exists a generic constant $M>0$, such that when $\lambda<M(1-\kappa)$, \eqref{KKT1} is true. Therefore, the first conclusion is true.

Under conditions (C2)--(C4), the second conclusion directly follows from the first conclusion.
\end{proof}

\begin{proof}[Proof of Theorem~\ref{error}]
We first show the first conclusion.
Define $\hat Y(\btheta_{2},\ldots,\btheta_{K})$ as the prediction by the Bayes rule and $\hat Y(\hat\btheta_{2},\ldots,\hat\btheta_{K})$ as the prediction as the prediction by the estimated classification rule. Also define $l_k=(\bX-\dfrac{\bmu_k}{2})^\T\btheta_k+\log(\pi_k)$ and $\hat l_k=(\bX-\dfrac{\hat\bmu_k}{2})^\T\hat\btheta_k+\log(\hat\pi_k)$.

Define $C(\epsilon)=\{|\hat\pi_k-\pi_k|\le \min\{\min_{k}\pi_k/2,\epsilon\}\}$. By the Bernstein inequality we have that $\Pr(C(\epsilon))\le C\exp(-Cn)$.

Assume that the event $A(\epsilon)\cap B(\epsilon)\cap C(\epsilon)$ has happened. By Lemma~\ref{prob.bounds}, we have
\begin{equation}
\Pr(A(\epsilon)\cap B(\epsilon)\cap C(\epsilon))\ge 1-Cpd\exp(-Cn\dfrac{\epsilon^2}{Kd^2})-CK\exp(-C\dfrac{n}{K^2})-Cp(K-1)\exp(-Cn\dfrac{\epsilon^2}{K})
\end{equation}

For any $\epsilon_0>0$,
\bean
R_n-R&\le& \Pr(\hat Y(\btheta_{2},\ldots,\btheta_{K})\ne \hat Y(\hat\btheta_{2},\ldots,\hat\btheta_{K}))\\
&\le&1-\Pr(|\hat l_k-l_k|<\epsilon_0/2,|l_k-l_{k'}|>\epsilon_0, \mbox{for any $k,k'$})\\
&\le&\Pr(|\hat l_k-l_k|\ge\epsilon_0/2 \mbox{ for some $k$})+\Pr(|l_k-l_{k'}|\le\epsilon_0 \mbox{ for some $k,k'$}).
\eean
Now, for $\bX$ in each class, $l_k-l_{k'}$ is normal with variance $(\btheta_k-\btheta_{k'})^\T\bSigma(\btheta_k-\btheta_{k''})$. Therefore,
\bean
\Pr(|l_k-l_{k'}|\le\epsilon_0 \mbox{ for some $k,k'$})&\le&\sum_{k^{''}}\Pr(|l_k-l_{k'}|\le\epsilon_0\mid Y=k^{''})\pi_{k^{''}}\\
&\le& \sum_{k,k^{'},k^{''}} \pi_{k^{''}} \dfrac{C\epsilon_0}{\{(\btheta_k-\btheta_{k'})^\T\bSigma(\btheta_k-\btheta_{k'})\}^{1/2}}\\
&\le&CK^2\epsilon_0.
\eean

On the other hand, conditional on training data, $\hat l_k-l_k$ is normal with mean $u(k,k')=\bmu_{k'}^\T(\hat\bmu_k-\bmu_k)+\dfrac{1}{2}(\bmu_k^\T\btheta_k-\hat\bmu_k^\T\hat\btheta_{k})+\log{\hat\pi_k}-\log{\pi_k}$ and variance $(\hat\btheta_k-\btheta_k)^\T\bSigma(\hat\btheta_k-\btheta_k)$ within class $k'$. By Markov's inequality, we have
\bean
\Pr(|\hat l_k-l_k|\ge\epsilon_0/2 \mbox{ for some $k$})&=&\sum_{k'}\pi_{k'}\Pr(|\hat l_k-l_k|\ge\epsilon_0/2\mid Y=k') \nonumber \\
&\le &CE\{\dfrac{\max_{k}(\hat\btheta_k-\btheta_k)^\T\bSigma(\hat\btheta_k-\btheta_k)}{(\epsilon_0-u(k,k'))^2}\}.
\eean
Moreover, under the event $A(\epsilon)\cap B(\epsilon)\cap C(\epsilon)$
\bean
\max_{k}(\hat\btheta_k-\btheta_k)^\T\bSigma(\hat\btheta_k-\btheta_k)&\le& C\lambda\\
|u(k,k')|&\le& |\bmu_{k'}(\hat\btheta_k-\btheta_k)|+\dfrac{1}{2}|\bmu_k^\T(\hat\bmu_k-\bmu_k)|\\
&&+\dfrac{1}{2}|(\bmu_k-\hat\bmu_k)^\T\hat\btheta_{k}|+|\log{\hat\pi_k}-\log{\pi_k}|\\
&\le& C_1\lambda\label{th2:eq1}
\eean
Hence, pick $\epsilon_0=M_2\lambda^{1/3}$ such that $\epsilon_0\ge C_1\lambda/2$, for $C_1$ in \eqref{th2:eq1}. Then $\Pr(|\hat l_k-l_k|\ge\epsilon_0/2 \mbox{ for some $k$})\le C\lambda^{1/3}$. It follows that $|R_n-R|\le M_1\lambda^{1/3}$ for some positive constant $M_1$.

Under Conditions (C2)--(C4), the second conclusion is a direct consequence of the first conclusion.
\end{proof}

We need the result in the following proposition to show Lemma~\ref{prop:fisher}. A slightly different version of the proposition has been presented in \cite{Fukunaga} (Pages 446-450), but we include the proof here for completeness.

\begin{proposition}\label{prop:fisher}
The solution to \eqref{fisher} consists of all the right eigenvectors of $\bSigma^{-1}\bSigma_b$ corresponding to positive eigenvalues.

\end{proposition}

\begin{proof}
For any $\bmeta_k$, set $\bu_k=\bSigma^{1/2}\bmeta_k$. It follows that solving \eqref{fisher} is equivalent to finding
\begin{equation}
(\bu_1^*,\ldots,\bu_{K-1}^*)=\arg\max_{\bu_k} \bu_k^\T\bSigma^{-1/2}\bdelta_0\bdelta_0^\T\bSigma^{-1/2}\bu_k, \mbox{ s.t. $\bu_k^\T\bu_k=1$ and $\bu_k^\T\bu_l=0$ for any $l<k$.}
\end{equation}
and then setting $\bmeta_k=\bSigma^{-1/2}\bu_k^*$. It is easy to see that $u_1^*,\ldots,u_{K-1}^*$ are the eigenvectors corresponding to positive eigenvalues of $\bSigma^{-1/2}\bdelta_0\bdelta_0^\T\bSigma^{-1/2}$. By Proposition~\ref{eigen}, let $\bA=\bSigma^{-1/2}\bdelta_0\bdelta_0^\T$, and $\bB=\bSigma^{-1/2}$ and we have that $\bmeta$ consists of all the eigenvectors of $\bSigma^{-1}\bdelta_0\bdelta_0^\T$ corresponding to positive eigenvalues.
\end{proof}

\begin{proposition}\label{eigen} (\cite{mardia}, Page 468, Theorem A.6.2) For two matrices $\bA$ and $\bB$, if $\bx$ is a non-trivial eigenvector of $\bA\bB$ for a nonzero eigenvalue, then $\by=\bB\bx$ is a non-trivial eigenvector of $\bB\bA$.

\end{proposition}

\begin{proof}[Proof of Lemma~\ref{lemma:fisher}]
 Set $\tilde{\bdelta}=(0_p,\bdelta)$ and $\bdelta_0=(\bmu_1-\bar\bmu,\ldots,\bmu_{K}-\bar\bmu)$. Note that $\bdelta \mathrm{1}_K=\sum_{k=2}^{K}\bmu_k-(K-1)\bmu_1=K(\bar\bmu-\bmu_1)$. Therefore, $\bdelta_0=\tilde \bdelta-\frac{1}{K}\tilde\bdelta\mathrm{1}_{K}\mathrm{1}_K^\T=\tilde\bdelta(\bI_K-\frac{1}{K}1_K1_K^\T)=\tilde\bdelta\bPi$. 

Then, since $\btheta_0=\bSigma^{-1}\tilde\bdelta$, we have $\btheta_0\bPi=\bSigma^{-1}\bdelta_0$ and $\btheta_0\bPi\bdelta_0^\T=\bSigma^{-1}\bdelta_0\bdelta_0^\T$. By Proposition~\ref{prop:fisher}, we have the desired conclusion.
\end{proof}

\bibliographystyle{agsm}
\bibliography{ref}

\end{document}